\documentclass[submission,copyright,creativecommons]{eptcs}
\providecommand{\event}{QPL 2022}
\pdfoutput=1

\usepackage{amsmath}
\usepackage{amssymb}
\usepackage{amsthm}
\usepackage{bbold}
\usepackage[utf8]{inputenc}
\usepackage[english]{babel}
\usepackage[T1]{fontenc}
\usepackage{hyperref}
\usepackage{relsize}
\usepackage[numbers,sort&compress]{natbib}
\usepackage{tikz}
\usetikzlibrary{
    arrows,
    shapes,
    decorations,
    intersections,
    backgrounds,
    positioning,
    circuits.ee.IEC
    }

	\tikzset{
	  rectangle with rounded corners north west/.initial=4pt,
	  rectangle with rounded corners south west/.initial=4pt,
	  rectangle with rounded corners north east/.initial=4pt,
	  rectangle with rounded corners south east/.initial=4pt,
	}
	\makeatletter
	\pgfdeclareshape{rectangle with rounded corners}{
	  \inheritsavedanchors[from=rectangle] 
	  \inheritanchorborder[from=rectangle]
	  \inheritanchor[from=rectangle]{center}
	  \inheritanchor[from=rectangle]{north}
	  \inheritanchor[from=rectangle]{south}
	  \inheritanchor[from=rectangle]{west}
	  \inheritanchor[from=rectangle]{east}
	  \inheritanchor[from=rectangle]{north east}
	  \inheritanchor[from=rectangle]{south east}
	  \inheritanchor[from=rectangle]{north west}
	  \inheritanchor[from=rectangle]{south west}
	  \backgroundpath{
	    \southwest \pgf@xa=\pgf@x \pgf@ya=\pgf@y
	    \northeast \pgf@xb=\pgf@x \pgf@yb=\pgf@y
	    \pgfkeysgetvalue{/tikz/rectangle with rounded corners north west}{\pgf@rectc}
	    \pgfsetcornersarced{\pgfpoint{\pgf@rectc}{\pgf@rectc}}
	    \pgfpathmoveto{\pgfpoint{\pgf@xa}{\pgf@ya}}
	    \pgfpathlineto{\pgfpoint{\pgf@xa}{\pgf@yb}}
	    \pgfkeysgetvalue{/tikz/rectangle with rounded corners north east}{\pgf@rectc}
	    \pgfsetcornersarced{\pgfpoint{\pgf@rectc}{\pgf@rectc}}
	    \pgfpathlineto{\pgfpoint{\pgf@xb}{\pgf@yb}}
	    \pgfkeysgetvalue{/tikz/rectangle with rounded corners south east}{\pgf@rectc}
	    \pgfsetcornersarced{\pgfpoint{\pgf@rectc}{\pgf@rectc}}
	    \pgfpathlineto{\pgfpoint{\pgf@xb}{\pgf@ya}}
	    \pgfkeysgetvalue{/tikz/rectangle with rounded corners south west}{\pgf@rectc}
	    \pgfsetcornersarced{\pgfpoint{\pgf@rectc}{\pgf@rectc}}
	    \pgfpathclose
	 }
	}
	\makeatother

	\tikzset{->-/.style={decoration={markings,mark=at position #1 with {\arrow{>}}},postaction={decorate}}}
	\tikzset{-<-/.style={decoration={markings,mark=at position #1 with {\arrow{<}}},postaction={decorate}}}

	\tikzstyle{every picture}=[baseline=-0.25em,scale=0.5]
	\pgfdeclarelayer{edgelayer}
	\pgfdeclarelayer{nodelayer}
	\pgfsetlayers{edgelayer,nodelayer,main}

	\tikzstyle{box} = [draw,shape=rectangle,inner sep=2pt,minimum height=6mm,minimum width=6mm,fill=white] 
	\tikzstyle{boxlarge} = [draw,shape=rectangle,inner sep=2pt,minimum height=1.5cm,minimum width=8mm,fill=white] 
	\tikzstyle{boxLarge} = [draw,shape=rectangle,inner sep=2pt,minimum height=2cm,minimum width=10mm,fill=white] 
	\tikzstyle{boxsmall} = [draw,shape=rectangle,inner sep=2pt,minimum height=3mm,minimum width=3mm,fill=white] 
	\tikzstyle{dot} = [inner sep=0mm,minimum width=3mm,minimum height=3mm,draw,shape=circle,text depth=-0.1mm,fill=white]
	\tikzstyle{Zbwdot} = [dot, fill=\Zbwcolour]
	\tikzstyle{Xbwdot} = [dot, fill=\Xbwcolour]
	\tikzstyle{Ybwdot} = [dot, fill=\Ybwcolour]
	\tikzstyle{Wbwdot} = [dot, fill=\Wbwcolour]
	\tikzstyle{antipode} = [boxsmall] 

	\tikzstyle{state} = [draw, rectangle with rounded corners,
	  rectangle with rounded corners north west=8pt,
	  rectangle with rounded corners south west=8pt,
	  rectangle with rounded corners north east=0pt,
	  rectangle with rounded corners south east=0pt,
	,inner sep=2pt,minimum height=6mm,minimum width=6mm,fill=white]
	\tikzstyle{stateV} = [draw, rectangle with rounded corners,
	  rectangle with rounded corners north west=0pt,
	  rectangle with rounded corners south west=8pt,
	  rectangle with rounded corners north east=0pt,
	  rectangle with rounded corners south east=8pt,
	,inner sep=2pt,minimum height=6mm,minimum width=6mm,fill=white]
	\tikzstyle{effect} = [draw, rectangle with rounded corners,
	  rectangle with rounded corners north west=0pt,
	  rectangle with rounded corners south west=0pt,
	  rectangle with rounded corners north east=8pt,
	  rectangle with rounded corners south east=8pt,
	,inner sep=2pt,minimum height=6mm,minimum width=6mm,fill=white]
	\tikzstyle{effectV} = [draw, rectangle with rounded corners,
	  rectangle with rounded corners north west=8pt,
	  rectangle with rounded corners south west=0pt,
	  rectangle with rounded corners north east=8pt,
	  rectangle with rounded corners south east=0pt,
	,inner sep=2pt,minimum height=6mm,minimum width=6mm,fill=white]
	\tikzstyle{scalar}=[diamond,draw,inner sep=1pt,font=\small,fill=white]

	\definecolor{themecolour}{RGB}{255,89,95}

	\tikzstyle{cdnode}=[fill=white]
	\tikzstyle{labelnode}=[fill=white]
	\tikzstyle{tightlabelnode}=[fill=white,inner sep = 0.1mm]
	\tikzstyle{none}=[inner sep=0pt]
	\tikzstyle{whiteline}=[-, line width=4pt, draw=white]

	\tikzstyle{trace}=[circuit ee IEC,thick,ground,scale=2.5]
	\tikzstyle{cotrace}=[circuit ee IEC,thick,ground,rotate=180,scale=2.5]
	\tikzstyle{upground}=[circuit ee IEC,thick,ground,rotate=90,scale=2.5,draw=themecolour]
	\tikzstyle{downground}=[circuit ee IEC,thick,ground,rotate=-90,scale=2.5,draw=themecolour]

	\tikzstyle{single} = [line width=1pt, draw=themecolour] %
	\tikzstyle{doubled} = [line width=1.8pt, draw=themecolour] 
	
	\tikzstyle{empty diagram}=[draw=gray!40!white,dashed,shape=rectangle,minimum width=1cm,minimum height=1cm]

\newcounter{theorem_c}
\newtheorem{theorem}[theorem_c]{Theorem}
\newtheorem{corollary}[theorem_c]{Corollary}
\newtheorem{proposition}[theorem_c]{Proposition}
\newtheorem{lemma}[theorem_c]{Lemma}
\theoremstyle{plain}

\newcommand{\fHilbCategory}{{\operatorname{fHilb}}}
\newcommand{\CPMCategory}[1]{\operatorname{CPM}\left(#1\right)}
\newcommand{\reals}{\mathbb{R}}
\newcommand{\id}{\mathbb{1}}

\title{Finite-dimensional Quantum Observables are the Special Symmetric $\dagger$-Frobenius Algebras of CP Maps}
\author{
    Stefano Gogioso
    \institute{Hashberg Ltd, London, UK}
    \institute{Quantum Group, University of Oxford, UK}
    \email{stefano.gogioso(at)cs.ox.ac.uk}
}
\def\titlerunning{Special Symmetric $\dagger$-Frobenius Algebras of CP maps}
\def\authorrunning{S. Gogioso}

\begin{document}

\maketitle

\begin{abstract}
\noindent We use \emph{purity}, a principle borrowed from the foundations of quantum information, to show that all special symmetric $\dagger$-Frobenius algebras in $\CPMCategory{\fHilbCategory}$---and, in particular, all classical structures---are \emph{canonical}, i.e. that they arise by doubling of special symmetric $\dagger$-Frobenius algebras in $\fHilbCategory$.
\end{abstract}

\section{Introduction}

The exact correspondence \cite{vicary2011categorical} between finite-dimensional C*-algebras and special symmetric $\dagger$-Frobenius algebras ($\dagger$-SSFAs) in $\fHilbCategory$---the dagger compact category of finite-dimensional complex Hilbert spaces and linear maps---is a cornerstone result in the categorical treatment of quantum theory \cite{abramsky2004categorical,coecke2017pqp,heunen2019categories}. A corresponding characterisation in $\CPMCategory{\fHilbCategory}$---the dagger compact category of finite-dimensional Hilbert spaces and completely positive maps \cite{selinger2007dagger}---has been an open question for around 10 years \cite{heunen2012completely,heunen2011mathoverflow}---of interest, for example, in the investigation of the robustness of sequentialisable quantum protocols \cite{boixo2012entangled,heunen2012completely}.

In this work we use \emph{purity}, a principle borrowed from the foundations of quantum information \cite{chiribella2010probabilistic}, to answer this question once and for all: the $\dagger$-SSFAs in $\CPMCategory{\fHilbCategory}$ are exactly the \emph{canonical} ones, the ones arising by doubling of $\dagger$-SSFAs in $\fHilbCategory$.
This is a notable result in that it allows fundamental notions from quantum theory---notably, those of quantum observable and measurement---to be defined directly in the diagrammatic language of CP maps, without reference to $\fHilbCategory$ or the CPM construction, and without relying on the biproduct or convex-linear structure of $\CPMCategory{\fHilbCategory}$.

\section{Purity}

The very definition of morphisms in the category $\CPMCategory{\fHilbCategory}$ means that a CP map $\Phi: \mathcal{H} \rightarrow \mathcal{K}$ can always be \emph{purified}, i.e. that it can be written in terms of a \emph{pure} map $\Psi: \mathcal{H} \rightarrow \mathcal{K}\otimes \mathcal{E}$ and discarding of an ``environment'' system $\mathcal{E}$:
\[
\begin{tikzpicture}
    \begin{pgfonlayer}{nodelayer}
        \node (Phi)  {$\Phi$};
        \node (eq) [right of = Phi] {$:=$};
        \node[box,doubled,minimum width = 10mm, xshift = +4mm] (Psi) [right of = eq] {$\Psi$};
        \node (Psi_in) [below of = Psi] {};
        \node (Psi_out) [above of = Psi, xshift = -3mm] {};
        \node (Psi_out_a) [below of = Psi_out] {};
        \node (Psi_out_trace) [above of = Psi, xshift = +4mm, yshift = +4mm,upground] {};
        \node (Psi_out_trace_a) [below of = Psi_out_trace] {};
    \end{pgfonlayer}
    \begin{pgfonlayer}{edgelayer}
        \draw[-,doubled] (Psi_in.center) to (Psi);
        \draw[-,doubled] (Psi_out.center) to (Psi_out_a.center);
        \draw[-,doubled] (Psi_out_trace) to (Psi_out_trace_a.center);
    \end{pgfonlayer}
\end{tikzpicture} 
\]
In this work, by a \emph{pure} CP map $\Psi:\mathcal{H} \rightarrow \mathcal{K}$ we mean a CP map in the form $\Psi=\CPMCategory{\psi}$, arising by doubling of a morphism $\psi:\mathcal{H} \rightarrow \mathcal{K}$ in $\fHilbCategory$. This ``diagrammatic'' notion of purity is connected to the notion of purity used in the foundations of quantum information by the following \emph{purity principle}.

\newpage
\begin{proposition}[Purity Principle]\hfill\\
    If the following holds for some pure CP maps $\Psi: \mathcal{H} \rightarrow \mathcal{K}\otimes \mathcal{E}$ and $F: \mathcal{H} \rightarrow \mathcal{K}$:
    \[
        \begin{tikzpicture}
            \begin{pgfonlayer}{nodelayer}
                \node[box,doubled,minimum width = 10mm] (Psi) [right of = eq] {$\Psi$};
                \node (Psi_in) [below of = Psi] {};
                \node (Psi_out) [above of = Psi, xshift = -3mm] {};
                \node (Psi_out_a) [below of = Psi_out] {};
                \node (Psi_out_trace) [above of = Psi, xshift = +4mm, yshift = +4mm,upground] {};
                \node (Psi_out_trace_a) [below of = Psi_out_trace] {};
                \node (eq) [right of = Psi] {$=$};
                \node[box,doubled] (F) [right of = eq]  {$F$};
                \node (F_in) [below of = F] {};
                \node (F_out) [above of = F] {};
            \end{pgfonlayer}
            \begin{pgfonlayer}{edgelayer}
                \draw[-,doubled] (F_in.center) to (F);
                \draw[-,doubled] (F_out.center) to (F);
                \draw[-,doubled] (Psi_in.center) to (Psi);
                \draw[-,doubled] (Psi_out.center) to (Psi_out_a.center);
                \draw[-,doubled] (Psi_out_trace) to (Psi_out_trace_a.center);
            \end{pgfonlayer}
        \end{tikzpicture} 
    \]
    then there is a normalised pure state $f$ on $\mathcal{E}$ such that:
    \[
        \begin{tikzpicture}
            \begin{pgfonlayer}{nodelayer}
                \node[box,doubled,minimum width = 10mm] (Psi) [right of = eq] {$\Psi$};
                \node (Psi_in) [below of = Psi] {};
                \node (Psi_out) [above of = Psi, xshift = -3mm] {};
                \node (Psi_out_a) [below of = Psi_out] {};
                \node (Psi_out_trace) [above of = Psi, xshift = +3mm] {};
                \node (Psi_out_trace_a) [below of = Psi_out_trace] {};
                \node (eq) [right of = Psi] {$=$};
                \node[box,doubled] (F) [right of = eq]  {$F$};
                \node (F_in) [below of = F] {};
                \node (F_out) [above of = F] {};
                \node[stateV,doubled] (f) [right of = F] {$f$};
                \node (f_out) [above of = f] {};
            \end{pgfonlayer}
            \begin{pgfonlayer}{edgelayer}
                \draw[-,doubled] (Psi_in.center) to (Psi);
                \draw[-,doubled] (Psi_out.center) to (Psi_out_a.center);
                \draw[-,doubled] (Psi_out_trace.center) to (Psi_out_trace_a.center);
                \draw[-,doubled] (F_in.center) to (F);
                \draw[-,doubled] (F_out.center) to (F);
                \draw[-,doubled] (f_out.center) to (f);
            \end{pgfonlayer}
        \end{tikzpicture} 
    \]
\end{proposition}
\noindent By expanding the discarding map in terms of some orthonormal basis of pure states, one straightforwardly gets an equivalent formulation of the principle in terms of sums.
\begin{proposition}[Purity Principle, sums version]\hfill\\
    If the following holds for pure CP maps $(\Psi_i: \mathcal{H} \rightarrow \mathcal{K})_i$ and $F: \mathcal{H} \rightarrow \mathcal{K}$:
    \[
        \begin{tikzpicture}
            \begin{pgfonlayer}{nodelayer}
                \node[box,doubled,minimum width = 8mm] (Psi) [right of = eq] {$\Psi_i$};
                \node (Psi_in) [below of = Psi] {};
                \node (Psi_out) [above of = Psi] {};
                \node (Psi_out_a) [below of = Psi_out] {};
                \node (sum) [left of = Psi, xshift = -2mm,yshift=-1.8mm] {$\mathlarger\sum_i$};
                \node (eq) [right of = Psi] {$=$};
                \node[box,doubled] (F) [right of = eq]  {$F$};
                \node (F_in) [below of = F] {};
                \node (F_out) [above of = F] {};
            \end{pgfonlayer}
            \begin{pgfonlayer}{edgelayer}
                \draw[-,doubled] (F_in.center) to (F);
                \draw[-,doubled] (F_out.center) to (F);
                \draw[-,doubled] (Psi_in.center) to (Psi);
                \draw[-,doubled] (Psi_out.center) to (Psi_out_a.center);
            \end{pgfonlayer}
        \end{tikzpicture} 
    \]
    then there are coefficients $p_i \in \reals^+$ summing to $1$ such that the following holds for all $i$:
    \[
        \begin{tikzpicture}
            \begin{pgfonlayer}{nodelayer}
                \node[box,doubled,minimum width = 8mm] (Psi) [right of = eq] {$\Psi_i$};
                \node (Psi_in) [below of = Psi] {};
                \node (Psi_out) [above of = Psi] {};
                \node (Psi_out_a) [below of = Psi_out] {};
                \node (eq) [right of = Psi] {$=$};
                \node (p) [right of = eq, xshift = -5mm] {$p_i$};
                \node[box,doubled] (F) [right of = eq, xshift = +2mm]  {$F$};
                \node (F_in) [below of = F] {};
                \node (F_out) [above of = F] {};
            \end{pgfonlayer}
            \begin{pgfonlayer}{edgelayer}
                \draw[-,doubled] (F_in.center) to (F);
                \draw[-,doubled] (F_out.center) to (F);
                \draw[-,doubled] (Psi_in.center) to (Psi);
                \draw[-,doubled] (Psi_out.center) to (Psi_out_a.center);
            \end{pgfonlayer}
        \end{tikzpicture} 
    \]
\end{proposition}
\noindent Because $\CPMCategory{\fHilbCategory}$ is dagger compact, finally, one can also straightforwardly obtain a more general formulation of the principle which replaces the purifying state $f$ with a purifying operator $f$.
\begin{proposition}[Purity Principle, operator version]\hfill\\
    If the following holds for some pure CP maps $\Psi: \mathcal{H} \otimes \mathcal{F} \rightarrow \mathcal{K}\otimes \mathcal{E}$ and $F: \mathcal{H} \rightarrow \mathcal{K}$:
    \[
        \begin{tikzpicture}
            \begin{pgfonlayer}{nodelayer}
                \node[box,doubled,minimum width = 10mm] (Psi) [right of = eq] {$\Psi$};
                \node (Psi_in) [below of = Psi, xshift = -3mm] {};
                \node (Psi_in_a) [above of = Psi_in] {};
                \node (Psi_in_trace) [below of = Psi, xshift = +4mm, yshift = -4mm,downground] {};
                \node (Psi_in_trace_a) [above of = Psi_in_trace] {};
                \node (Psi_out) [above of = Psi, xshift = -3mm] {};
                \node (Psi_out_a) [below of = Psi_out] {};
                \node (Psi_out_trace) [above of = Psi, xshift = +4mm, yshift = +4mm,upground] {};
                \node (Psi_out_trace_a) [below of = Psi_out_trace] {};
                \node (eq) [right of = Psi] {$=$};
                \node[box,doubled] (F) [right of = eq]  {$F$};
                \node (F_in) [below of = F] {};
                \node (F_out) [above of = F] {};
            \end{pgfonlayer}
            \begin{pgfonlayer}{edgelayer}
                \draw[-,doubled] (F_in.center) to (F);
                \draw[-,doubled] (F_out.center) to (F);
                \draw[-,doubled] (Psi_in.center) to (Psi_in_a.center);
                \draw[-,doubled] (Psi_in_trace) to (Psi_in_trace_a.center);
                \draw[-,doubled] (Psi_out.center) to (Psi_out_a.center);
                \draw[-,doubled] (Psi_out_trace) to (Psi_out_trace_a.center);
            \end{pgfonlayer}
        \end{tikzpicture} 
    \]
    then there is a pure CP map $f: \mathcal{F} \rightarrow \mathcal{E}$ such that:
    \[
        \begin{tikzpicture}
            \begin{pgfonlayer}{nodelayer}
                \node[box,doubled,minimum width = 10mm] (Psi) [right of = eq] {$\Psi$};
                \node (Psi_in) [below of = Psi, xshift = -3mm] {};
                \node (Psi_in_a) [above of = Psi_in] {};
                \node (Psi_in_trace) [below of = Psi, xshift = +3mm] {};
                \node (Psi_in_trace_a) [above of = Psi_in_trace] {};
                \node (Psi_out) [above of = Psi, xshift = -3mm] {};
                \node (Psi_out_a) [below of = Psi_out] {};
                \node (Psi_out_trace) [above of = Psi, xshift = +3mm] {};
                \node (Psi_out_trace_a) [below of = Psi_out_trace] {};
                \node (eq) [right of = Psi] {$=$};
                \node[box,doubled] (F) [right of = eq]  {$F$};
                \node (F_in) [below of = F] {};
                \node (F_out) [above of = F] {};
                \node[box,doubled] (f) [right of = F]  {$f$};
                \node (f_in) [below of = f] {};
                \node (f_out) [above of = f] {};
            \end{pgfonlayer}
            \begin{pgfonlayer}{edgelayer}
                \draw[-,doubled] (Psi_in.center) to (Psi_in_a.center);
                \draw[-,doubled] (Psi_in_trace.center) to (Psi_in_trace_a.center);
                \draw[-,doubled] (Psi_out.center) to (Psi_out_a.center);
                \draw[-,doubled] (Psi_out_trace.center) to (Psi_out_trace_a.center);
                \draw[-,doubled] (F_in.center) to (F);
                \draw[-,doubled] (F_out.center) to (F);
                \draw[-,doubled] (f_in.center) to (f);
                \draw[-,doubled] (f_out.center) to (f);
            \end{pgfonlayer}
        \end{tikzpicture} 
    \]
\end{proposition}

\section{Isometries of CP maps}

\begin{theorem}[Purification of CP isometries]\label{thm_CPisometries}\hfill\\
    Every isometry $\Phi$ in $\CPMCategory{\fHilbCategory}$ can be written as a $\reals^+$-linear combination of pure isometries $V_i$ with pairwise orthogonal images. 
    \[
        \Phi^\dagger \circ \Phi = \id
        \hspace{1cm}\Rightarrow\hspace{1cm}
        \Phi = \sum_i q_i V_i
        \hspace{0.5cm} \text{with} \hspace{0.5cm}
        V_i^\dagger\circ V_j = \delta_{ij} \, \id 
    \]
    In pictures, this means that if the following equation holds, where $\Psi$ is any purification of $\Phi$:
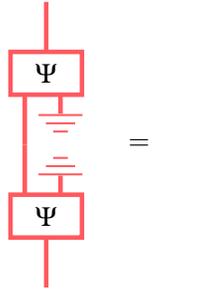
    \[
        \scalebox{0.95}{

        }
    \]
    then we must in fact have:
    \[
        %
 
    \]
    Furthermore, the coefficients satisfy $\sum_i (q_i)^2 = 1$, and they can be chosen to be all non-zero.
\end{theorem}
\begin{proof}
    By the existence of purifications, we can obtain $\Phi: \mathcal{H} \rightarrow \mathcal{K}$ by discarding an ``environment'' system $\mathcal{E}$ from some pure $\Psi: \mathcal{H} \rightarrow \mathcal{K}\otimes \mathcal{E}$:
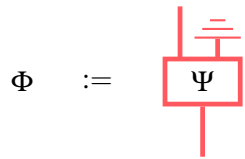
    \[
    %
 
    \]
    By the operator version of the purity principle, the isometry equation $\Phi^\dagger \circ \Phi = \id$ for $\Phi$ implies the following equation for $\Psi$, where $f$ is some pure CP map $\mathcal{E} \rightarrow \mathcal{E}$:
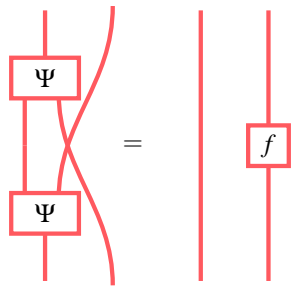
    \[
        \scalebox{0.9}{
        %
        }
    \]
    The pure CP map $f$ is self-adjoint, which means that we can find an orthonormal basis of pure states $(\varphi_i)_{i=1}^{\dim{\mathcal{E}}}$ and associated non-negative real numbers $(q_i)_{i}$ such that:
    \[
        %
 
    \]
    This way we get:
    \[
    %
 
    \]
    Furthermore, the $V_i$ pure maps we just defined are isometries with orthogonal images:
    \[
    %
 
    \]
    Finally, we have that:
    \[
    %
 
    \]
    so we conclude that $\sum_{i} (q_i)^2=1$.
\end{proof}

\section{Main result}

A comonoid $(\Delta,E)$ on an object $\mathcal{H}$ of the dagger compact category $\CPMCategory{\fHilbCategory}$ is a pair of completely positive maps satisfying the associativity and unit laws:
    \[
        %
 
    \]

\begin{theorem}[Ophidian isometric CP comonoids are pure]\label{thm_main}\hfill\\
    If $(\Delta, E)$ is a comonoid in $\CPMCategory{\fHilbCategory}$ which is \emph{isometric} and satisfies the \emph{snake equations}, then the CP maps $\Delta$ and $E$ are both pure.\footnote{Associativity is not actually necessary for this result to hold.}
    \[
        %
    \]
\end{theorem}
\begin{proof} By Theorem \ref{thm_CPisometries} and the isometry condition, we know that:
    \[
        %
       
    \]
    for some pure isometries $V_i$ of $\fHilbCategory$ with $V_i^\dagger \circ V_j = \delta_{ij}\,\id{}$. The coefficients $q_i$ satisfy $\underline{q}\cdot \underline{q} = \sum_i (q_i)^2 = 1$, and we choose them to be all non-zero. By the existence of purifications, we decompose $E$ as a sum of non-zero pure effects:
    \[
        %
   
    \]
    By the purity principle and the unit laws for the comonoid, we deduce that there are coefficients $r_{i,k},\,l_{i,k} \in \reals^+$ such that:
    \[
        %
   
    \]
    Writing $l_i := \sum_k l_{i,k}$ and $r_i := \sum_k r_{i,k}$, the \emph{unit laws} imply that $\underline{q}\cdot \underline{r} = \sum_i q_i r_i = 1 = \sum_i q_i l_i = \underline{q} \cdot \underline{l}$: it cannot therefore be that for all $i,k$ we have $l_{i,k} = 0$ or that for all $i,k$ we have $r_{i,k} = 0$. Picking some $i,k$ such that $l_{i,k} \neq 0$, we deduce that $E$ is, in fact, a non-zero pure effect:
    \[
        %
 
    \] 
    Because $E$ is a pure effect, we will henceforth drop the discarding map in our notation:
    \[
        %
    \]
    The same trick can then be used to show that $r_i = l_i$ for all $i$:
    \[
        %
 
    \] 
    Having established that the counit is pure, we now move on to establish that the comultiplication is pure as well.
    From the snake equations, the purity principle implies the existence of coefficients $\lambda_{i,j}, \rho_{i,j} \in \reals^+$ such that:
    \[
        %
    \]
    We now proceed to show that $\lambda_{i,j} = \delta_{ij}$, using the left snake equation. The proof that $\rho_{i,j} = \delta_{ij}$ is analogous, using the right snake equation.
    Taking the trace on both sides of the snake equation, we obtain the following equation, where $\mathcal{H}$ is the underlying Hilbert space for the comonoid (we get $\operatorname{dim}(\mathcal{H})^2$ because we are working in the CPM category, where the scalar is doubled):
    \[
        %
    \]
    Taking the trace on the LHS of the snake equation and using the fact that $V_j^\dagger V_i = \delta_{ij} \id{}$, we get a different equation:
    \[
        %
    \]
    The rightmost equality follows from isometry and the snake equation:
    \[
        %
    \]
    Having established that $\lambda_{i,j} = \delta_{ij} = \rho_{i,j}$, we now show that $n \geq 2$ leads to a contradiction unless $\operatorname{dim}(\mathcal{H}) = 0$ (in which case the statement of this theorem is trivial). Indeed, we have the following equation for all $i, j$:
    \[
        %
    \]
    If we have $n \geq 2$, then we can set $i \neq j$ in the equation above, concluding that the RHS state is the zero state.
    However, we also know that $V_i$ is an isometry and that the adjoint of the counit is a state of norm $\operatorname{dim}(\mathcal{H})^2$, hence the RHS state also has norm $\operatorname{dim}(\mathcal{H})^2$.
    So either $n = 1$, in which case the comultiplication is pure, or $n \geq 2 $ and $\operatorname{dim}(\mathcal{H})=0$, in which case the comultiplication is the zero map, which is also pure.
\end{proof}

\begin{lemma}
    \label{projective_algebras}
    Let $\delta: \mathcal{H} \rightarrow \mathcal{H} \otimes \mathcal{H}$ and $\epsilon: \mathcal{H} \rightarrow \mathbb{C}$ be morphisms in $\fHilbCategory$ which satisfy the associativity law up to a phase $e^{i\alpha}$, the symmetry law up to a phase $e^{i\sigma}$, the left unit law exactly, the right unit law up to a phase $e^{i\rho}$ and the Frobenius law up to a phase $e^{i\varphi}$:
    \[
    \begin{tikzpicture}[dot/.append style={rectangle,}]
        \begin{pgfonlayer}{nodelayer}
            \node (delta) [box,single,minimum width = 10mm]{$\delta$};
            \node (delta_in) [below of = delta] {};
            \node[box,single,minimum width = 10mm] (delta2) [above left of = delta, xshift = +2mm, yshift = 4mm] {$\delta$};
            \node (delta_outl) [above of = delta, xshift = -10mm, yshift = 10mm] {}; 
            \node (delta_outc) [above of = delta, xshift = 0mm, yshift = 10mm] {};
            \node (delta_outr) [above of = delta, xshift = +8mm, yshift = 10mm] {};
        \end{pgfonlayer}
        \begin{pgfonlayer}{edgelayer}
            \draw[-,single] (delta_in) to (delta);
            \draw[-,single,out=270,in=45] (delta_outr.center) to (delta);
            \draw[-,single,out=270,in=135] (delta2) to (delta);
            \draw[-,single,out=270,in=135] (delta_outl.center) to (delta2);
            \draw[-,single,out=270,in=45] (delta_outc.center) to (delta2);
        \end{pgfonlayer}
        \node (eq) [right of = delta, xshift = 10mm, yshift=6mm] {$\stackrel{a}{=}$};
        \begin{pgfonlayer}{nodelayer}
            \node (r) [right of =eq, xshift = 0mm, yshift=0mm] {$e^{i\alpha}$};
            \node[box,single,minimum width = 10mm] [right of = eq, xshift = 10mm, yshift=-6mm] (delta) {$\delta$};
            \node (delta_in) [below of = delta] {};
            \node[box,single,minimum width = 10mm] (delta2) [above right of = delta, xshift = -2mm, yshift = 4mm] {$\delta$};
            \node (delta_outl) [above of = delta, xshift = -8mm, yshift = 10mm] {}; 
            \node (delta_outc) [above of = delta, xshift = 0mm, yshift = 10mm] {};
            \node (delta_outr) [above of = delta, xshift = +10mm, yshift = 10mm] {};
        \end{pgfonlayer}
        \begin{pgfonlayer}{edgelayer}
            \draw[-,single] (delta_in) to (delta);
            \draw[-,single,out=270,in=135] (delta_outl.center) to (delta);
            \draw[-,single,out=270,in=45] (delta2) to (delta);
            \draw[-,single,out=270,in=135] (delta_outc.center) to (delta2);
            \draw[-,single,out=270,in=45] (delta_outr.center) to (delta2);
        \end{pgfonlayer}

        \node (spacer) [right of = eq, xshift = 30mm] {};

        \begin{pgfonlayer}{nodelayer}
            \node[box,single,minimum width = 10mm] [right of = spacer, xshift = 10mm] (V) {$\delta$};
            \node (q) [left of = V, xshift = 0mm, yshift = 0mm] {};
            \node (V_in) [below of = V] [stateV, single] {$\epsilon^\dagger$};
            \node (V_outl) [above of = V, xshift =-3mm, yshift = 5mm] {};
            \node (V_outl_a) [below of = V_outl, yshift = -3mm] {};
            \node (V_outr) [above of = V, xshift =+3mm, yshift = 5mm] {};
            \node (V_outr_a) [below of = V_outr, yshift = -3mm] {};
        \end{pgfonlayer}
        \begin{pgfonlayer}{edgelayer}
            \draw[-,single] (V_in.center) to (V);
            \draw[-,single,out=270,in=90] (V_outl.center) to (V_outr_a);
            \draw[-,single,out=270,in=90] (V_outr.center) to (V_outl_a);
        \end{pgfonlayer}
        \node (eq) [right of = V, xshift = +10mm] {$\stackrel{s}{=}$};
        \begin{pgfonlayer}{nodelayer}
            \node (r) [right of =eq, xshift = 0mm] {$e^{i\sigma}$};
            \node[box,single,minimum width = 10mm] [right of = eq, xshift=10mm] (V) {$\delta$};
            \node (V_in) [below of = V] [stateV, single] {$\epsilon^\dagger$};
            \node (V_outl) [above of = V, xshift =-3mm, yshift = 5mm] {};
            \node (V_outl_a) [below of = V_outl, yshift = -3mm] {};
            \node (V_outr) [above of = V, xshift =+3mm, yshift = 5mm] {};
            \node (V_outr_a) [below of = V_outr, yshift = -3mm] {};
        \end{pgfonlayer}
        \begin{pgfonlayer}{edgelayer}
            \draw[-,single] (V_in.center) to (V);
            \draw[-,single] (V_outl.center) to (V_outl_a);
            \draw[-,single] (V_outr.center) to (V_outr_a);
        \end{pgfonlayer}
    \end{tikzpicture}
    \]
    \[
        \begin{tikzpicture}[dot/.append style={rectangle,}]
            \begin{pgfonlayer}{nodelayer}
                \node[box,single,minimum width = 10mm] (V)  {$\delta$};
                \node (V_in) [below of = V] {};
                \node[effectV,single] (V_outl) [above of = V, xshift =-3mm]  {$\epsilon$};       
                \node (V_outl_a) [below of = V_outl] {};
                \node (V_outr) [above of = V, xshift =+3mm] {};
                \node (V_outr_a) [below of = V_outr] {};
            \end{pgfonlayer}
            \begin{pgfonlayer}{edgelayer}
                \draw[-,single] (V_in.center) to (V);
                \draw[-,single] (V_outl.center) to (V_outl_a);
                \draw[-,single] (V_outr.center) to (V_outr_a);
            \end{pgfonlayer}    
            \node (eq) [right of = V, xshift = 5mm] {$\stackrel{lu}{=}$};
            \node (l) [right of =eq, xshift = -3mm] {$$};
            \node (id) [right of = eq, xshift = +3mm] {};
            \node (id_in) [above of = id] {};
            \node (id_out) [below of = id] {};
            \draw[-,single] (id_in.center) to (id_out.center);
            \node (spacer) [right of = id, xshift = 22mm] {};
            \begin{pgfonlayer}{nodelayer}
                \node[box,single,minimum width = 10mm] (V) [right of = spacer, xshift = 10mm]  {$\delta$};
                \node (V_in) [below of = V] {};
                \node (V_outl) [above of = V, xshift =-3mm]  {};    
                \node (V_outl_a) [below of = V_outl] {};
                \node[effectV,single] (V_outr) [above of = V, xshift =+3mm] {$\epsilon$};
                \node (V_outr_a) [below of = V_outr] {};
            \end{pgfonlayer}
            \begin{pgfonlayer}{edgelayer}
                \draw[-,single] (V_in.center) to (V);
                \draw[-,single] (V_outl.center) to (V_outl_a);
                \draw[-,single] (V_outr.center) to (V_outr_a);
            \end{pgfonlayer}    
            \node (eq) [right of = V, xshift = 5mm] {$\stackrel{ru}{=}$};
            \node (r) [right of =eq, xshift = -1mm] {$e^{i\rho}$};
            \node (id) [right of = eq, xshift = +3mm] {};
            \node (id_in) [above of = id] {};
            \node (id_out) [below of = id] {};
            \draw[-,single] (id_in.center) to (id_out.center);
            \node (spacer) [right of = id, xshift = -9mm] {};
        \end{tikzpicture}   
    \]
    \[
        \begin{tikzpicture}[dot/.append style={rectangle,}]
            \begin{pgfonlayer}{nodelayer}
                \node (delta)[box,single,minimum width = 10mm]  {$\delta$};
                \node (delta_in) [below of = delta, yshift=0mm] {};
                \node (delta2) [above right of = delta, yshift = 13mm, xshift = 5mm] [box,single,minimum width = 10mm]{$\delta^\dagger$};
                \node (delta_out) [above of = delta2, yshift=0mm] {}; 
                \node (delta_inr) [below of = delta2, xshift = 8mm, yshift = -20mm] {};
                \node (delta_outl) [above of = delta, xshift = -8mm, yshift = 20mm] {}; 
            \end{pgfonlayer}
            \begin{pgfonlayer}{edgelayer}
                \draw[-,single] (delta_in) to (delta);
                \draw[-,single,out=45,in=-45, looseness=1.2] (delta) to (delta2);
                \draw[-,single,out=135,in=-90] (delta) to (delta_outl);
                \draw[-,single,out=90,in=-135] (delta_inr) to (delta2);
                \draw[-,single] (delta2) to (delta_out);
            \end{pgfonlayer}
            \node (eq) [right of = delta, xshift = 20mm, yshift = 9mm] {$\stackrel{f}{=}$};
            \begin{pgfonlayer}{nodelayer}
                \node (r) [right of =eq, xshift = 0mm] {$e^{i\varphi}$};
                \node (delta)[right of = eq, xshift = 20mm, yshift = -9mm] [box,single,minimum width = 10mm] {$\delta$};
                \node (delta_in) [below of = delta, yshift=0mm] {};
                \node (delta2) [above left of = delta, yshift = 13mm, xshift = -5mm] [box,single,minimum width = 10mm]{$\delta^\dagger$};
                \node (delta_out) [above of = delta2, yshift=0mm] {}; 
                \node (delta_inr) [below of = delta2, xshift = -8mm, yshift = -20mm] {};
                \node (delta_outl) [above of = delta, xshift = 8mm, yshift = 20mm] {}; 
            \end{pgfonlayer}
            \begin{pgfonlayer}{edgelayer}
                \draw[-,single] (delta_in) to (delta);
                \draw[-,single,out=135,in=-135, looseness=1.2] (delta) to (delta2);
                \draw[-,single,out=45,in=-90] (delta) to (delta_outl);
                \draw[-,single,out=90,in=-45] (delta_inr) to (delta2);
                \draw[-,single] (delta2) to (delta_out);
            \end{pgfonlayer}
        \end{tikzpicture}
    \]
    Then $(\delta,\epsilon,\delta^\dagger,\epsilon^\dagger)$ is a symmetric $\dagger$-Frobenius algebra in $\fHilbCategory$, i.e. $\alpha = \rho = \sigma = \varphi = 0 \text{ mod } 2\pi$.
\end{lemma}
\begin{proof}
    The associativity law and left unit law prove that $\alpha = 0 \text{ mod }2\pi$:
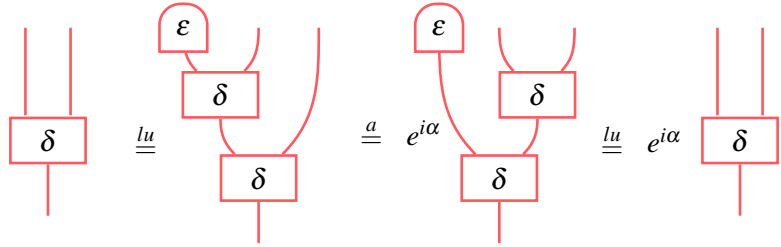
    \[
    \begin{tikzpicture}[dot/.append style={rectangle,}]
        \begin{pgfonlayer}{nodelayer}
            \node[box,single,minimum width = 10mm] (V) {$\delta$};
            \node (V_in) [below of = V] {};
            \node (V_outl) [above of = V, xshift =-3mm, yshift = 5mm] {};
            \node (V_outl_a) [below of = V_outl, yshift = -3mm] {};
            \node (V_outr) [above of = V, xshift =+3mm, yshift = 5mm] {};
            \node (V_outr_a) [below of = V_outr, yshift = -3mm] {};
        \end{pgfonlayer}
        \begin{pgfonlayer}{edgelayer}
            \draw[-,single] (V_in.center) to (V);
            \draw[-,single] (V_outl.center) to (V_outl_a);
            \draw[-,single] (V_outr.center) to (V_outr_a);
        \end{pgfonlayer}
        \node (eq) [right of = V, xshift = 3mm] {$\stackrel{lu}{=}$};
        \begin{pgfonlayer}{nodelayer}
            \node [right of =eq, xshift = 5mm, yshift=-5mm] (delta) [box,single,minimum width = 10mm]{$\delta$};
            \node (delta_in) [below of = delta] {};
            \node[box,single,minimum width = 10mm] (delta2) [above left of = delta, xshift = +2mm, yshift = 4mm] {$\delta$};
            \node [effectV,single] (delta_outl) [above of = delta, xshift = -10mm, yshift = 10mm] {$\epsilon$}; 
            \node (delta_outc) [above of = delta, xshift = 0mm, yshift = 10mm] {};
            \node (delta_outr) [above of = delta, xshift = +8mm, yshift = 10mm] {};
        \end{pgfonlayer}
        \begin{pgfonlayer}{edgelayer}
            \draw[-,single] (delta_in) to (delta);
            \draw[-,single,out=270,in=45] (delta_outr.center) to (delta);
            \draw[-,single,out=270,in=135] (delta2) to (delta);
            \draw[-,single,out=270,in=135] (delta_outl.center) to (delta2);
            \draw[-,single,out=270,in=45] (delta_outc.center) to (delta2);
        \end{pgfonlayer}
        \node (eq) [right of = delta, xshift = 5mm, yshift=6mm] {$\stackrel{a}{=}$};
        \begin{pgfonlayer}{nodelayer}
            \node (r) [right of =eq, xshift = -3mm, yshift=0mm] {$e^{i\alpha}$};
            \node[box,single,minimum width = 10mm] [right of = eq, xshift = 7mm, yshift=-6mm] (delta) {$\delta$};
            \node (delta_in) [below of = delta] {};
            \node[box,single,minimum width = 10mm] (delta2) [above right of = delta, xshift = -2mm, yshift = 4mm] {$\delta$};
            \node [effectV,single] (delta_outl) [above of = delta, xshift = -8mm, yshift = 10mm] {$\epsilon$}; 
            \node (delta_outc) [above of = delta, xshift = 0mm, yshift = 10mm] {};
            \node (delta_outr) [above of = delta, xshift = +10mm, yshift = 10mm] {};
        \end{pgfonlayer}
        \begin{pgfonlayer}{edgelayer}
            \draw[-,single] (delta_in) to (delta);
            \draw[-,single,out=270,in=135] (delta_outl.center) to (delta);
            \draw[-,single,out=270,in=45] (delta2) to (delta);
            \draw[-,single,out=270,in=135] (delta_outc.center) to (delta2);
            \draw[-,single,out=270,in=45] (delta_outr.center) to (delta2);
        \end{pgfonlayer}
        \node (eq) [right of = delta, xshift = 5mm, yshift=5mm] {$\stackrel{lu}{=}$};
        \begin{pgfonlayer}{nodelayer}
            \node (r) [right of =eq, xshift = -3mm] {$e^{i\alpha}$};
            \node[box,single,minimum width = 10mm] [right of = eq, xshift=7mm] (V) {$\delta$};
            \node (V_in) [below of = V] {};
            \node (V_outl) [above of = V, xshift =-3mm, yshift = 5mm] {};
            \node (V_outl_a) [below of = V_outl, yshift = -3mm] {};
            \node (V_outr) [above of = V, xshift =+3mm, yshift = 5mm] {};
            \node (V_outr_a) [below of = V_outr, yshift = -3mm] {};
        \end{pgfonlayer}
        \begin{pgfonlayer}{edgelayer}
            \draw[-,single] (V_in.center) to (V);
            \draw[-,single] (V_outl.center) to (V_outl_a);
            \draw[-,single] (V_outr.center) to (V_outr_a);
        \end{pgfonlayer}
    \end{tikzpicture}
    \]
    Then, the associativity law and the unit laws prove that $\rho = 0 \text{ mod }2\pi$:
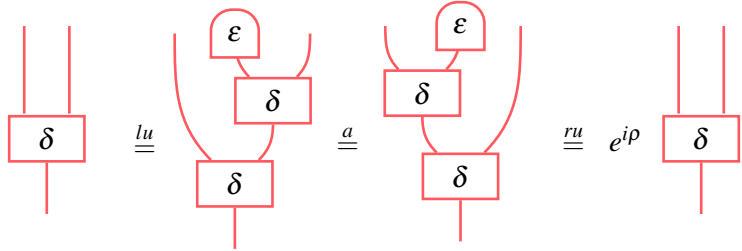
    \[
    \begin{tikzpicture}[dot/.append style={rectangle,}]
        \begin{pgfonlayer}{nodelayer}
            \node[box,single,minimum width = 10mm] (V) {$\delta$};
            \node (V_in) [below of = V] {};
            \node (V_outl) [above of = V, xshift =-3mm, yshift = 5mm] {};
            \node (V_outl_a) [below of = V_outl, yshift = -3mm] {};
            \node (V_outr) [above of = V, xshift =+3mm, yshift = 5mm] {};
            \node (V_outr_a) [below of = V_outr, yshift = -3mm] {};
        \end{pgfonlayer}
        \begin{pgfonlayer}{edgelayer}
            \draw[-,single] (V_in.center) to (V);
            \draw[-,single] (V_outl.center) to (V_outl_a);
            \draw[-,single] (V_outr.center) to (V_outr_a);
        \end{pgfonlayer}
        \node (eq) [right of = V, xshift = 3mm] {$\stackrel{lu}{=}$};
        \begin{pgfonlayer}{nodelayer}
            \node[box,single,minimum width = 10mm] [right of = eq, xshift = 2mm, yshift=-6mm] (delta) {$\delta$};
            \node (delta_in) [below of = delta] {};
            \node[box,single,minimum width = 10mm] (delta2) [above right of = delta, xshift = -2mm, yshift = 4mm] {$\delta$};
            \node (delta_outl) [above of = delta, xshift = -8mm, yshift = 10mm] {}; 
            \node [effectV,single] (delta_outc) [above of = delta, xshift = 0mm, yshift = 10mm] {$\epsilon$};
            \node (delta_outr) [above of = delta, xshift = +10mm, yshift = 10mm] {};
        \end{pgfonlayer}
        \begin{pgfonlayer}{edgelayer}
            \draw[-,single] (delta_in) to (delta);
            \draw[-,single,out=270,in=135] (delta_outl.center) to (delta);
            \draw[-,single,out=270,in=45] (delta2) to (delta);
            \draw[-,single,out=270,in=135] (delta_outc.center) to (delta2);
            \draw[-,single,out=270,in=45] (delta_outr.center) to (delta2);
        \end{pgfonlayer}
        \node (eq) [right of = delta, xshift = 5mm, yshift=6mm] {$\stackrel{a}{=}$};
        \begin{pgfonlayer}{nodelayer}
            \node [right of =eq, xshift = 5mm, yshift=-5mm] (delta) [box,single,minimum width = 10mm]{$\delta$};
            \node (delta_in) [below of = delta] {};
            \node[box,single,minimum width = 10mm] (delta2) [above left of = delta, xshift = +2mm, yshift = 4mm] {$\delta$};
            \node (delta_outl) [above of = delta, xshift = -10mm, yshift = 10mm] {}; 
            \node [effectV,single] (delta_outc) [above of = delta, xshift = 0mm, yshift = 10mm] {$\epsilon$};
            \node (delta_outr) [above of = delta, xshift = +8mm, yshift = 10mm] {};
        \end{pgfonlayer}
        \begin{pgfonlayer}{edgelayer}
            \draw[-,single] (delta_in) to (delta);
            \draw[-,single,out=270,in=45] (delta_outr.center) to (delta);
            \draw[-,single,out=270,in=135] (delta2) to (delta);
            \draw[-,single,out=270,in=135] (delta_outl.center) to (delta2);
            \draw[-,single,out=270,in=45] (delta_outc.center) to (delta2);
        \end{pgfonlayer}
        \node (eq) [right of = delta, xshift = 5mm, yshift=5mm] {$\stackrel{ru}{=}$};
        \begin{pgfonlayer}{nodelayer}
            \node (r) [right of =eq, xshift = -3mm] {$e^{i\rho}$};
            \node[box,single,minimum width = 10mm] [right of = eq, xshift=7mm] (V) {$\delta$};
            \node (V_in) [below of = V] {};
            \node (V_outl) [above of = V, xshift =-3mm, yshift = 5mm] {};
            \node (V_outl_a) [below of = V_outl, yshift = -3mm] {};
            \node (V_outr) [above of = V, xshift =+3mm, yshift = 5mm] {};
            \node (V_outr_a) [below of = V_outr, yshift = -3mm] {};
        \end{pgfonlayer}
        \begin{pgfonlayer}{edgelayer}
            \draw[-,single] (V_in.center) to (V);
            \draw[-,single] (V_outl.center) to (V_outl_a);
            \draw[-,single] (V_outr.center) to (V_outr_a);
        \end{pgfonlayer}
    \end{tikzpicture}
    \]
    Then, the symmetry law and the unit laws prove that $\sigma = 0 \text{ mod }2\pi$:
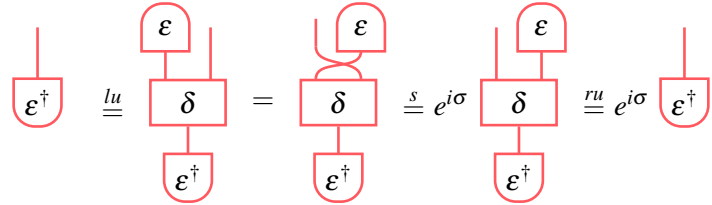
    \[
    \begin{tikzpicture}[dot/.append style={rectangle,}]
        \begin{pgfonlayer}{nodelayer}
            \node (V) [stateV, single] {$\epsilon^\dagger$};
            \node (V_out) [above of = V]  {};
        \end{pgfonlayer}
        \begin{pgfonlayer}{edgelayer}
            \draw[-,single] (V.center) to (V_out.center);
        \end{pgfonlayer} 
        \node (eq) [right of = V, xshift = 0mm] {$\stackrel{lu}{=}$};
        \begin{pgfonlayer}{nodelayer}
            \node[box,single,minimum width = 10mm] [right of=eq, xshift = 0mm] (V)  {$\delta$};
            \node (V_in) [stateV, single, below of = V] {$\epsilon^\dagger$};
            \node [effectV,single] (V_outl) [above of = V, xshift =-3mm]  {$\epsilon$};    
            \node (V_outl_a) [below of = V_outl] {};
            \node (V_outr) [above of = V, xshift =+3mm] {};
            \node (V_outr_a) [below of = V_outr] {};
        \end{pgfonlayer}
        \begin{pgfonlayer}{edgelayer}
            \draw[-,single] (V_in.center) to (V);
            \draw[-,single] (V_outl.center) to (V_outl_a);
            \draw[-,single] (V_outr.center) to (V_outr_a);
        \end{pgfonlayer}
        \node (eq) [right of = V, xshift = 0mm] {$=$};
        \begin{pgfonlayer}{nodelayer}
            \node[box,single,minimum width = 10mm] [right of=eq, xshift = 0mm] (V)  {$\delta$};
            \node (V_in) [stateV, single, below of = V] {$\epsilon^\dagger$};
            \node (V_outl) [above of = V, xshift =-3mm, yshift=-1.2mm]  {};    
            \node (V_outl2) [above of = V_outl, yshift=-9mm]  {};    
            \node (V_outl_a) [below of = V_outl, yshift=3mm] {};
            \node[effectV,single] (V_outr) [above of = V, xshift =+3mm] {$\epsilon$};
            \node (V_outr_a) [below of = V_outr, yshift=1.8mm] {};
        \end{pgfonlayer}
        \begin{pgfonlayer}{edgelayer}
            \draw[-,single] (V_in.center) to (V);
            \draw[-,single] (V_outl.south) to (V_outl2.north);
            \draw[-,single,out=270,in=90] (V_outl) to (V_outr_a);
            \draw[-,single,out=270,in=90] (V_outr) to (V_outl_a);
        \end{pgfonlayer}
        \node (eq) [right of = V, xshift = 0mm] {$\stackrel{s}{=}$};
        \begin{pgfonlayer}{nodelayer}
            \node (r) [right of =eq, xshift = -5mm] {$e^{i\sigma}$};
            \node[box,single,minimum width = 10mm] [right of=eq, xshift = 4mm] (V)  {$\delta$};
            \node (V_in) [stateV, single, below of = V] {$\epsilon^\dagger$};
            \node (V_outl) [above of = V, xshift =-3mm]  {};    
            \node (V_outl_a) [below of = V_outl] {};
            \node[effectV,single] (V_outr) [above of = V, xshift =+3mm] {$\epsilon$};
            \node (V_outr_a) [below of = V_outr] {};
        \end{pgfonlayer}
        \begin{pgfonlayer}{edgelayer}
            \draw[-,single] (V_in.center) to (V);
            \draw[-,single] (V_outl.center) to (V_outl_a);
            \draw[-,single] (V_outr.center) to (V_outr_a);
        \end{pgfonlayer}
        \node (eq) [right of = V, xshift = 0mm] {$\stackrel{ru}{=}$};
        \begin{pgfonlayer}{nodelayer}
            \node (r) [right of =eq, xshift = -5mm] {$e^{i\sigma}$};
            \node (V) [stateV, single] [right of=eq, xshift = 2mm] {$\epsilon^\dagger$};
            \node (V_out) [above of = V]  {};
        \end{pgfonlayer}
        \begin{pgfonlayer}{edgelayer}
            \draw[-,single] (V.center) to (V_out.center);
        \end{pgfonlayer} 
    \end{tikzpicture}
    \]
    Finally, the Frobenius law and the right unit law prove that $\varphi = 0 \text{ mod }2\pi$:
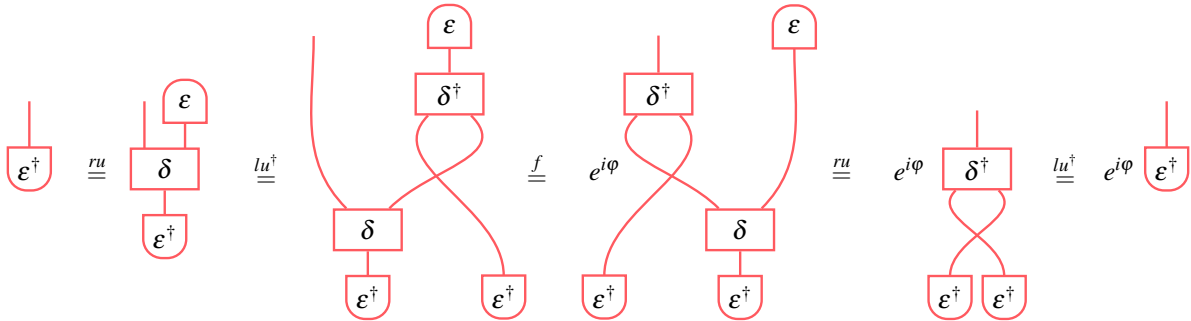
    \[
    \scalebox{0.9}{$
        \begin{tikzpicture}[dot/.append style={rectangle,}]
            \begin{pgfonlayer}{nodelayer}
                \node (V) [stateV, single] {$\epsilon^\dagger$};
                \node (V_out) [above of = V]  {};
            \end{pgfonlayer}
            \begin{pgfonlayer}{edgelayer}
                \draw[-,single] (V.center) to (V_out.center);
            \end{pgfonlayer} 
            \node (eq) [right of = V, xshift = 0mm] {$\stackrel{ru}{=}$};
            \begin{pgfonlayer}{nodelayer}
                \node[box,single,minimum width = 10mm] [right of=eq, xshift = 0mm] (V)  {$\delta$};
                \node (V_in) [stateV, single, below of = V] {$\epsilon^\dagger$};
                \node (V_outl) [above of = V, xshift =-3mm]  {};    
                \node (V_outl_a) [below of = V_outl] {};
                \node[effectV,single] (V_outr) [above of = V, xshift =+3mm] {$\epsilon$};
                \node (V_outr_a) [below of = V_outr] {};
            \end{pgfonlayer}
            \begin{pgfonlayer}{edgelayer}
                \draw[-,single] (V_in.center) to (V);
                \draw[-,single] (V_outl.center) to (V_outl_a);
                \draw[-,single] (V_outr.center) to (V_outr_a);
            \end{pgfonlayer} 
            \node (eq) [right of = V, xshift = 5mm] {$\stackrel{lu^\dagger}{=}$};
            \begin{pgfonlayer}{nodelayer}
                \node (delta)[box,single,minimum width = 10mm] [right of=eq, xshift = 5mm, yshift = -9mm] {$\delta$};
                \node (delta_in) [stateV, single, below of = delta, yshift=0mm] {$\epsilon^\dagger$};
                \node (delta2) [above right of = delta, yshift = 13mm, xshift = 5mm] [box,single,minimum width = 10mm]{$\delta^\dagger$};
                \node (delta_out) [effectV, single, above of = delta2, yshift=0mm] {$\epsilon$}; 
                \node (delta_inr) [stateV, single, below of = delta2, xshift = 8mm, yshift = -20mm] {$\epsilon^\dagger$};
                \node (delta_outl) [above of = delta, xshift = -8mm, yshift = 20mm] {}; 
            \end{pgfonlayer}
            \begin{pgfonlayer}{edgelayer}
                \draw[-,single] (delta_in) to (delta);
                \draw[-,single,out=45,in=-45, looseness=1.2] (delta) to (delta2);
                \draw[-,single,out=135,in=-90] (delta) to (delta_outl);
                \draw[-,single,out=90,in=-135] (delta_inr) to (delta2);
                \draw[-,single] (delta2) to (delta_out);
            \end{pgfonlayer}
            \node (eq) [right of = delta, xshift = 15mm, yshift = 9mm] {$\stackrel{f}{=}$};
            \begin{pgfonlayer}{nodelayer}
                \node (r) [right of =eq, xshift = 0mm] {$e^{i\varphi}$};
                \node (delta)[right of = eq, xshift = 20mm, yshift = -9mm] [box,single,minimum width = 10mm] {$\delta$};
                \node (delta_in) [stateV, single, below of = delta, yshift=0mm] {$\epsilon^\dagger$};
                \node (delta2) [above left of = delta, yshift = 13mm, xshift = -5mm] [box,single,minimum width = 10mm]{$\delta^\dagger$};
                \node (delta_out) [above of = delta2, yshift=0mm] {}; 
                \node (delta_inr) [stateV, single, below of = delta2, xshift = -8mm, yshift = -20mm] {$\epsilon^\dagger$};
                \node (delta_outl) [effectV, single, above of = delta, xshift = 8mm, yshift = 20mm] {$\epsilon$}; 
            \end{pgfonlayer}
            \begin{pgfonlayer}{edgelayer}
                \draw[-,single] (delta_in) to (delta);
                \draw[-,single,out=135,in=-135, looseness=1.2] (delta) to (delta2);
                \draw[-,single,out=45,in=-90] (delta) to (delta_outl);
                \draw[-,single,out=90,in=-45] (delta_inr) to (delta2);
                \draw[-,single] (delta2) to (delta_out);
            \end{pgfonlayer}
            \node (eq) [right of = delta, xshift = 5mm, yshift=9mm] {$\stackrel{ru}{=}$};
            \begin{pgfonlayer}{nodelayer}
                \node (r) [right of =eq, xshift = 0mm] {$e^{i\varphi}$};
                \node (delta) [right of = eq, xshift = 10mm] [box,single,minimum width = 10mm]{$\delta^\dagger$};
                \node (delta_out) [above of = delta, yshift=0mm] {};
                \node (delta_inl) [stateV, single, below of = delta, xshift=4mm, yshift=-9mm] {$\epsilon^\dagger$};
                \node (delta_inr) [stateV, single, below of = delta, xshift=-4mm, yshift=-9mm] {$\epsilon^\dagger$};
            \end{pgfonlayer}
            \begin{pgfonlayer}{edgelayer}
                \draw[-,single,out=90,in=-135] (delta_inl) to (delta);
                \draw[-,single,out=90,in=-45] (delta_inr) to (delta);
                \draw[-,single] (delta) to (delta_out);
            \end{pgfonlayer}
            \node (eq) [right of = delta, xshift = 3mm, yshift=0mm] {$\stackrel{lu^\dagger}{=}$};
            \begin{pgfonlayer}{nodelayer}
                \node (r) [right of =eq, xshift = -2mm] {$e^{i\varphi}$};
                \node (V) [stateV, single] [right of=eq, xshift = 5mm] {$\epsilon^\dagger$};
                \node (V_out) [above of = V]  {};
            \end{pgfonlayer}
            \begin{pgfonlayer}{edgelayer}
                \draw[-,single] (V.center) to (V_out.center);
            \end{pgfonlayer} 
        \end{tikzpicture}
    $}
    \]
    This concludes our proof.
\end{proof}

\begin{corollary}[CP $\dagger$-SSFAs are all canonical]\label{cor_main}\hfill\\
    The $\dagger$-SSFAs in $\CPMCategory{\fHilbCategory}$ are all canonical, i.e. they all arise by doubling of $\dagger$-SSFAs in $\fHilbCategory$.
\end{corollary}
\begin{proof}
    Now let $(\Delta,E,\Delta^\dagger,E^\dagger)$ be a $\dagger$-SSFA in $\CPMCategory{\fHilbCategory}$.
    Because $(\Delta, E)$ is a comonoid which is isometric and satisfies the snake equations, we now know that $\Delta$ and $E$ are pure, i.e. that we can find linear maps $\delta$ and $\epsilon$ in $\fHilbCategory$ such that $\Delta = \operatorname{CPM}(\delta)$ and $E = \operatorname{CPM}(\epsilon)$ arise by doubling.
    We now wish to conclude that $(\delta,\epsilon,\delta^\dagger,\epsilon^\dagger)$ form a $\dagger$-SSFA in $\fHilbCategory$, but this doesnt immediately follow from the equations in $\CPMCategory{\fHilbCategory}$: while $\delta$ is certainly an isometry in $\fHilbCategory$, the associativity law, unit laws, symmetry law and Frobenius law for $\delta$ and $\epsilon$ are only guaranteed to hold up to phase.
    In fact, $(\delta,\epsilon,\delta^\dagger,\epsilon^\dagger)$ is not, in general, a $\dagger$-SSFA in $\fHilbCategory$.
    However, it is easy to show (cf. Lemma \ref{projective_algebras} below) that $(\delta,e^{-i\lambda}\epsilon,\delta^\dagger,e^{i\lambda}\epsilon^\dagger)$ is always a $\dagger$-SSFA in $\fHilbCategory$, where $e^{i\lambda}$ is the phase associated to the identity in the left unit law.
    Because $\operatorname{CPM}(e^{-i\lambda}\epsilon) = \operatorname{CPM}(\epsilon)$, this is enough to prove our result.
\end{proof}

\section*{Acknowledgements}

The author would like to thank Chris Heunen \cite{heunen2012completely,heunen2011mathoverflow} and Sergio Boixo \cite{heunen2012completely} for formulating the original question.
The author would also like to thank an anonymous QPL 2022 reviewer for thoroughly checking the correctness of the proofs, as well as suggesting several improvements to the presentation of this work.

\bibliographystyle{eptcs}
\bibliography{bibliography}

\end{document}